\documentclass[11pt]{article}

\usepackage{tikz}
\usepackage{bm}
\usepackage[normalem]{ulem}
\usepackage{parskip}
\usepackage{amsmath}
\usepackage{amssymb}
\usepackage{xspace}
\usepackage{algorithm}
\usepackage[noend]{algpseudocode}

\usepackage{graphicx} 
\usepackage{stmaryrd}
\usepackage{scalerel}
\usepackage{mathtools}
\usepackage{mathrsfs}
\usepackage{thmtools}
\usepackage{xspace}
\usepackage{hyperref}
\usepackage{tikz}
\usepackage{sectsty}

\newcommand{\sff}[1]{\llfloor #1 \rrfloor}
\newcommand{\TWM}{\textsc{Twm}\xspace}
\newcommand{\WAM}{\textsc{Wam}\xspace}

\newcommand{\SAM}{\textsc{Sam}\xspace}

\newcommand{\BIM}{\textsc{Bim}\xspace}
\newcommand{\WNM}{\textsc{OWNM}\xspace}
\newcommand{\TGM}{\textsc{Tgm}\xspace}
\newcommand{\bimrho}{0.2862}
\newcommand{\RRM}{\textsc{Rrm}\xspace}
\newcommand{\ALG}{\textsc{ALG}\xspace}
\newcommand{\OPT}{\textsc{OPT}\xspace}
\newcommand{\bbC}{\mathbb{C}}
\newcommand{\bbR}{\mathbb{R}}
\DeclareMathOperator{\Ex}{\mathbb{E}}

\newcounter{theorem} 
\newtheorem{xdefinition}[theorem]{Definition}
\newtheorem{xobservation}[theorem]{Observation}
\newtheorem{xtheorem}[theorem]{Theorem}
\newtheorem{xlemma}[theorem]{Lemma}     
\newtheorem{xproposition}[theorem]{Proposition}
\newtheorem{xcorollary}[theorem]{Corollary}
{\hspace*{\fill}\raisebox{-1pt}{\boldmath$\Box$}\end{xdefinition}}
{\hspace*{\fill}\raisebox{-1pt}{\boldmath$\Box$}\end{xobservation}}
\newenvironment{theorem}{\begin{xtheorem}\rm}{\end{xtheorem}}
\newenvironment{lemma}{\begin{xlemma}\rm}{\end{xlemma}}
\newenvironment{proposition}{\begin{xproposition}\rm}{\end{xproposition}}
\newenvironment{corollary}{\begin{xcorollary}\rm}{\end{xcorollary}}
\newenvironment{proof}{\begin{trivlist}\item[]{\bf Proof }}%
{\hspace*{\fill}\raisebox{-1pt}{\boldmath$\Box$}\end{trivlist}}

\begin{document}
\input{figs}

\title{On the Online Weighted Non-Crossing \\ Matching Problem\thanks{Boyar and Larsen were supported in part by the Danish Council for Independent Research grants DFF-0135-00018B and DFF-4283-00079B. Kamali was supported in part by the Natural Science and Engineering Research Council of Canada grant DGECR-2018-00059. Pankratov was supported in part by the Natural Science and Engineering Research Council of Canada.}}

\author{\mbox{}\hspace*{-1.5em}\begin{tabular}{ll}
        \begin{tabular}{l}
          Joan Boyar \\
          {\small\textsc{University of Southern Denmark}} \\
          \texttt{joan@imada.sdu.dk}
        \end{tabular}
        &
        \begin{tabular}{l}
          Shahin Kamali \\
          {\small\textsc{York University}} \\
          \texttt{kamalis@york.ca}
        \end{tabular}
        \\[5ex]
        \begin{tabular}{l}
          Kim S. Larsen \\
          {\small\textsc{University of Southern Denmark}} \\
          \texttt{kslarsen@imada.sdu.dk}
        \end{tabular}
        &
        \begin{tabular}{l}
          Ali Fata Lavasani \\
          {\small\textsc{Huawei, Canada}} \\
          \texttt{fata.lavasani@gmail.com}
        \end{tabular} \\[5ex]
        \begin{tabular}{l}
          Yaqiao Li \\
          {\small\textsc{Shenzhen University}} \\
          {\small\textsc{of Advanced Technology}} \\
          \texttt{liyaqiao@suat-sz.edu.cn}
        \end{tabular}
        &
        \begin{tabular}{l}
          Denis Pankratov \\
          {\small\textsc{Concordia University}} \\
          \texttt{denis.pankratov@concordia.ca}
        \end{tabular}
        \end{tabular}}

\date{March 5, 2026}

\maketitle

\begin{abstract}
    We introduce and study the weighted version of an online matching problem in the Euclidean plane with non-crossing constraints: points with non-negative weights arrive online, and an algorithm can match an arriving point to one of the unmatched previously arrived points. In the classic model, the decision on how to match (if at all) a newly arriving point is irrevocable. The goal is to maximize the total weight of matched points under the constraint that straight-line segments corresponding to the edges of the matching do not intersect. The unweighted version of the problem was introduced in the offline setting by Atallah in 1985, and this problem became a subject of study in the online setting with and without advice in several recent papers.

    We observe that deterministic online algorithms cannot guarantee a non-trivial competitive ratio for the weighted problem, but we give upper and lower bounds on the problem with bounded weights. In contrast to the deterministic case, we show that using randomization, a constant competitive ratio is possible for arbitrary weights. We also study other variants of the problem, including revocability and collinear points, both of which permit non-trivial online algorithms, and we give upper and lower bounds for the attainable competitive ratios. Finally, we prove an advice complexity bound for obtaining optimality, improving the best known bound.
\end{abstract}

\section{Introduction}

We introduce and study the following problem, which we call Online Weighted Non-Crossing Matching (\WNM).
Suppose $2n$ points $p_1, \ldots, p_{2n}$ in the Euclidean plane arrive online one-by-one. When $p_i$ arrives, its positive weight $w(p_i)\in \mathbb{R}_{> 0}$ is revealed and an algorithm has the option of matching $p_i$ to one of the unmatched previously revealed points, or leave $p_i$ unmatched. In the classic online model, the decisions of the algorithm are irrevocable. There is a non-crossing constraint, requiring that the straight-line segments corresponding to the edges of the matching do not intersect. Assuming that the points are in general position,\footnote{\,The phrase ``general position'' is a much used term in computational geometry, used when one wants to indicate that one is not interested in treating special cases, usually because this is cumbersome and uninteresting and clutters the central parts of an algorithm. One special case for our problem that we would want to disregard for some of the variants we treat would be (three or more) collinear points.} the goal is to design an algorithm that maximizes the sum of the weights of matched points.

The interest in geometric settings, particularly the Euclidean plane setting, for the matching problem stems from applications in image processing~\cite{cohen1999finding} and circuit board design~\cite{hershberger1997efficient}. In such applications, one is often required to construct a matching between various geometric shapes, such as rectangles or circles, representing vertices, using straight-line segments or, more generally, curves. Geometry enters the picture due to constraints on the edges, such as avoiding intersections among the edges, as well as avoiding edge-vertex intersections. These constraints can have a significant impact on the offline complexity of the problem.

The unweighted version of the Non-Crossing Matching problem (when $w(p_i) = 1$ for all $i \in \{1, \ldots, 2n\}$) has been studied both in the offline setting~\cite{atallah1985matching,HS1992Alg} and the online setting~\cite{bose2020non,sajadpour2021non,Kamali2022randomizedNM,lavasani2023}. We review the history of the problem in detail in Section~\ref{sec:related}. For now, it suffices to observe that an offline algorithm that knows the locations of all the points in advance can match all the points while satisfying the non-crossing constraint.
This is the notion of $\OPT$ as used in competitive analysis, and one can think of it as the best possible offline algorithm when there is no limitations on its computational complexity.
Thus, the value of $\OPT$ is always $W := \sum_{i = 1}^{2n} w(p_i)$. The performance of an online algorithm is measured by its competitive ratio, which for our problem corresponds to the fraction of $W$ that the algorithm can guarantee to achieve in the worst-case. 

It is relatively easy to see that when there are no restrictions on the weights of points, no deterministic online algorithm can be competitive, i.e., obtain a competitive ratio greater than zero (in particular, this is an immediate corollary of Theorem~\ref{th:mainLo}). We study different regimes under which the problem admits algorithms achieving non-trivial competitive ratios.

\section{Related Work}
\label{sec:related}

Given $2n$ points in $\bbR^2$ in general position, the basic non-crossing matching (NM) problem is to find a non-crossing matching with the largest possible number of edges. Observe that the minimum-length Euclidean matching is non-crossing\footnote{The proof is by contradiction. Take a smallest total length perfect matching. If it has a crossing, then look at the induced quadrilateral -- the length of the diagonals is larger than the length of a pair of opposite sides. Now, replace the diagonals by the shorter pair of sides. This new perfect matching has smaller total length, contradicting our choice of the perfect matching, so the smallest total length matching cannot have any crossings.}. Hence, a perfect NM always exists. The NM problem and its variants have been extensively studied in the offline setting. Hershberger and Suri~\cite{HS1992Alg} gave an algorithm that finds a perfect NM in time $\Theta(n\log n)$. First Atallah~\cite{atallah1985matching} and then Dumitrescu and Steiger~\cite{dumitrescu2000matching} gave efficient algorithms for the bichromatic version of the problem, where the points are divided into two subsets, and matching edges can only be formed between the two subsets. Other versions of the problem considered in the literature include requiring the NM to be stable~\cite{Ruangwises2019StableNM,Hamada2021StableNM2}, requiring two NMs to be compatible (edges in two NMs are also non-crossing if they only share endpoints)~\cite{Aichholzer2009compatibleNM}, and requiring compatible NMs to satisfy a certain diversity constraint~\cite{Misra2022diverseNM}.

Several studies considered optimization problems over all NMs. The objective functions include maximizing the sum of the Euclidean lengths of the matching edges~\cite{Alon1993longNM}, minimizing the length of the longest matching edge~\cite{Abu2014bottleneck}, and other similar combinations of minimization and maximization~\cite{Mantas2024Variants}.

Another line of research is to relax the non-crossing constraint and allow certain crossings. An important problem is to understand the size of a crossing family, that is, matching edges that are pairwise crossing. A breakthrough by Pach et al.~\cite{Pach2019crossingfamily} showed that the largest crossing family has linear size. Aichholzer et al.~\cite{Aichholzer2023kcrossing} studied the counting problem of $k$-crossing matchings.

At least two works~\cite{Balogh2005WNM,SU2011heavyNM} considered weighted NM on $n$ points, where every point has weight in $\{1,2,\ldots, n\}$. Balogh et al.~\cite{Balogh2005WNM} considered the weight of an edge to be the sum of the weights of the two endpoints modulo $n$, and studied the typical size of an NM with distinct edge weights (this is called non-crossing harmonic matching). Sakai and Urrutia~\cite{SU2011heavyNM} considered the weight of an edge to be the minimum weight of the two endpoints, and they studied lower and upper bounds for the maximum weighted NM.

In pure mathematics, NM has been studied as a tool to understand the representation theory of groups~\cite{Armstrong2013Rep2,Patrias2022Rep}. A tuple of NMs (a so-called necklace) satisfying a specific property is used to study the topology of harmonic algebraic curves associated with a polynomial over $\bbC$~\cite{Savitt2009polynomials_necklace}. 
Extremal graph problems where an NM of size $k$ plays the role of a forbidden subgraph are studied in~\cite{Aichholzer2010NMasSubgraph,Gyarfas2011Extremal}.

Besides the application in image processing and circuit design, as mentioned in the introduction, 
NM has also found other applications. One major application is in computational biology. A restricted version of NM (the points are all on a circle), and $k$-non-crossing matchings (no $k$ pairwise intersecting edges, which reduces to the standard NM when $k=2$) have been studied to understand RNA structures~\cite{Bafna1995RNA, Chen2009RNA,Vladimirov2013RNA}. 
In applications that are related to visibility problems (such as robotics) and geometric shape matching, one replaces all or a subset of the points in question by geometric objects~\cite{Rappaport2002visibilitymatching,Aloupis2015objects}. For example, when the question is to match objects to objects, then an edge $(p,q)$ between two objects $A$ and $B$ can be formed by choosing arbitrary points $p$ from $A$ and $q$ from $B$, conditioned on the edge $(p,q)$ not crossing other objects.

The online NM has only been studied very recently. Bose et al.~\cite{bose2020non} initiated the study of online (unweighted) NM and showed that the competitive ratio of deterministic algorithms is $2/3$, while Kamali et al.~\cite{Kamali2022randomizedNM} gave a randomized algorithm that, in expectation, matches approximately a fraction $0.6695$ of all points. The online bichromatic NM has also been studied in~\cite{bose2020non,sajadpour2021non}.
Finally, the advice complexity was studied in~\cite{bose2020non,lavasani2023}. In particular, Lavasani and Pankratov~\cite{lavasani2023} resolved the advice complexity of solving online bichromatic NM optimally on a circle and gave a lower bound of $n/3-1$ and an upper bound of $3n$ on the advice complexity of online NM in the plane.

\section{Contributions}
\label{sec:contributions}

The results we obtain for algorithmic and problem variants of the problem
Online Weighted Non-Crossing Matching (\WNM)
can be summarized as follows:

\begin{itemize}
    \item Deterministic algorithms with bounded weights: In the Restricted \WNM, we assume that the weights of points are restricted to the interval $[L, U]$ for some $L\le U$, known to the algorithm at the beginning of the execution. Note that by scaling, we can assume without loss of generality that $L = 1$, i.e., all weights lie in $[1,U]$. We show that the competitive ratio of any deterministic online algorithm is $O{\left(2^{-\sqrt{\log U}}\right)}$ (Theorem~\ref{th:mainLo}). We also present a deterministic online algorithm, Wait-and-Match (\WAM), which has competitive ratio $\Omega\left(2^{-2\sqrt{\log U}}\right)$ (Theorem~\ref{th:mainUp}).

    \item Randomized algorithms: We show, perhaps surprisingly, that randomization alone is enough to guarantee a constant competitive ratio for arbitrary weights. We present a simple randomized online algorithm, called Tree-Guided-Matching (\TGM), and prove that it has a competitive ratio of $1/3$ (Theorem~\ref{thm:rand_ub}). We supplement this result by showing that no randomized online algorithm can achieve a competitive ratio better than $16/17$, even for the unweighted version of the problem (Theorem~\ref{thm:rand_lb}).

    \item Revocable setting: We show that allowing revocable acceptances (see the beginning of Section~\ref{sec:revoke} for the definition of the model) permits one to obtain a competitive ratio of approximately $\bimrho$ by a deterministic algorithm, even when the weights are unrestricted (Theorem~\ref{thm:revoke_ub}). We supplement this result by showing that no deterministic algorithm with revoking can achieve a competitive ratio better than $2/3$ (Theorem~\ref{thm:revoke_negative}).

    \item Collinear points: We also study this problem when the points are not in general position but are all located on a line (see Section \ref{sec:collinear}). We show that, even in the unweighted case, neither revoking nor randomization alone helps achieve a non-trivial competitive ratio. However, we present a randomized algorithm with revoking that achieves a competitive ratio of at least $0.5$ in the unweighted version.
    
    \item Advice complexity: We present a new algorithm, called Split-and-Match (\SAM), that uses $\lceil \log C_n \rceil<2n$ bits of advice  (see beginning of Section~\ref{sec:advice} for the definition of the model) to achieve optimality (Theorem~\ref{thm:sam}), where $C_n$ is the $n^\text{th}$ Catalan number. This improves upon the previously known bound of $3n$ on the advice complexity of the problem~\cite{lavasani2023}. Since \SAM achieves a perfect matching, it does not matter whether the given points are weighted or not.
\end{itemize}

\section{Preliminaries}

The input to the matching problems considered in this work is an online sequence $I = \langle p_1,\ldots, p_{2n}\rangle$ of points \emph{in general position}, where $p_i$ has a positive real-valued weight $w(p_i) \in \mathbb{R}_{>0}$. We use $W$ to denote the total weight of all the points, i.e., $W = \sum_{i = 1}^{2n} w(p_i)$. For the Restricted \WNM, the weights are assumed to lie in the interval $[1,U]$ for some known value of $U$, which is a part of the problem definition, and, thus, not a part of the input. Upon the arrival of $p_i$, an online algorithm must either leave it unmatched or match it with an unmatched point $p_j$ ($j<i$), provided that the line segment between $p_i$ and $p_j$, denoted by $\overline{p_ip_j}$, does not cross the line segments between previously matched pairs of points. The objective is to maximize the total weight of matched points. For an online algorithm, $\ALG$ (respectively, offline optimal algorithm $\OPT$), we use $\ALG(I)$  (respectively, $\OPT(I)$) to denote the total weight of points matched by the algorithm on input $I$.
When it is clear from the context, the symbol $\overline{pq}$ is also used to denote the full line passing through the two points $p$ and $q$, dividing a convex region into two sub-regions.

We say that a deterministic online algorithm $\ALG$ is \emph{$\rho$-competitive} if there exists a constant $c$ such that for every input sequence $I$ we have
\[\ALG(I) \ge \rho \cdot \OPT(I) - c.\]
For a randomized algorithm, the above inequality is replaced by the following
\[\mathbb{E}(\ALG(I)) \ge \rho \cdot \OPT(I) - c.\]
If $c=0$, then we call the competitive ratio $\rho$ \emph{strict}, and we say that $\ALG$ is \emph{strictly $\rho$-competitive}. If $c \neq 0$, then the ratio can be called asymptotic. Note that for the Restricted \WNM, we allow $c$ to depend on $U$ when considering asymptotic competitiveness. Thus, an algorithm achieving asymptotic competitive ratio $\rho$ is allowed to leave a constant number of points unmatched (regardless of their weights) beyond the $(1-\rho)$-fraction of $W$.

\subsection{Warm-Up: The Two-Weight Case}

As a warm-up, we first consider  the case where points can have only weight $1$ or weight $U>1$. Several of our later algorithms and analyses are generalizations of, or inspired by, the two-weight case.
Also, as we shall see, in this case our algorithm achieves a tight  competitive ratio of $1/3$ as $U\to \infty$. Incidentally, $1/3$  appears again as a bound on the competitive ratio for our randomized algorithm in Section \ref{sec:randomized}.

We consider the following algorithm that we call
``Two-Weight-Matching'' (\TWM).
Assume the points appear in a bounding box,~$B$. 
Throughout its execution, \TWM maintains a \emph{convex partitioning} of~$B$. Initially, there is only one region formed by the entire~$B$. The algorithm matches two points only if they appear in the same convex region. 
Whenever two points in a convex region $R$ are matched, the line segment between them is extended until it hits the boundary of $R$, which results in partitioning $R$ into two smaller convex regions.

When a point $p$ arrives in a region $R$ in which there are already some unmatched points, \TWM matches $p$ with another point in the following cases: (1) if $p$ has weight $U$, then $p$ matches with a point $q$ of the largest weight; (2) if $p$ has weight $1$ and there is a point of weight $U$ in $R$, $p$ matches with that point; (3) otherwise, $p$ and all unmatched points in $R$ have weight $1$, and $p$ matches with a point $q$ in $R$ if this can be done in such a way that none of the new convex regions end up empty.

{\bf Observation}: it is an invariant during the execution of \TWM that if a region contains more than one point, all the points have weight~$1$, and that no region contains more than three unmatched points.

\begin{theorem}
    For the two-weight OWNM with $U\geq 3$, the competitive ratio of \TWM is at least $1/3$.
\end{theorem}

\begin{proof}
    At the completion of~\TWM, let $m_1$ be the number of points of weight~$1$ that are matched, and $m_U$ be the number of points of weight~$U$ that are matched.  Thus, $\TWM = m_1 + Um_U$. Furthermore, let $r_1$ be the number of regions with at least one unmatched point of weight~$1$ and $r_U$ the number of regions with a single unmatched point of weight~$U$. By the observation above, these are all the possible regions containing unmatched points, and these two types of regions are disjoint.

Since it takes two matched points to divide a region into two,
$r_1+r_U \leq \frac{m_1+m_U}{2} + 1$,
so $r_1 \leq \frac{m_1+m_U}{2} + 1 - r_U$.

Assume that there is an unmatched point of weight~$U$ in region $R_1$ and let
$R_2$ be the other region that was created at the same time as $R_1$
by matching two points $p$ and $q$. If both $p$ and $q$ have weight~$1$, then both $R_1$ and $R_2$ contain at least one point of weight~$1$ and, therefore, $R_1$ cannot have an unmatched point of weight~$U$, which is a contradiction. So, either $p$ or $q$ has weight~$U$.
Thus, the match between these two points were made without regards to possibly leaving regions empty, so, following the match, points of weight~$U$ could arrive into empty regions $R_1$ and $R_2$. Hence, $r_U \leq 2 m_U$.

Since $\OPT$ matches all points,
\begin{align*}
\OPT & \leq m_1 + Um_U + 3r_1 + Ur_U \\[.5ex]
     & \leq  m_1 + Um_U + 3(\frac{m_1+m_U}{2} + 1 - r_U) + Ur_U, \mbox{ by the above} \\[.5ex]
     & =     m_1 + Um_U + 3(\frac{m_1+m_U}{2} + 1) + (U-3)r_U \\[.5ex]
     & \leq  m_1 + Um_U + 3(\frac{m_1+m_U}{2} + 1) + (U-3)2m_U, \mbox{ since $U\geq 3$} \\[.5ex]
     & \leq  \frac52 m_1 + 3Um_U + 3\frac{m_U}{2} + 3 - 6m_U \\[.5ex]
     & \leq  3m_1 + 3Um_U = 3 \TWM, \mbox{ since  $m_U\geq 1$.}    
\end{align*}
Thus, $\TWM \geq \frac13 \OPT - 1$, and the result follows.
\end{proof}

We now show that no algorithm is better than $1/3$-competitive, for
arbitrarily large~$U$.
\begin{theorem}
    For the two-weight OWNM, the competitive ratio of any deterministic online algorithm is at most $1/3+ 2/(3U+3)$.
\end{theorem}
\begin{proof}
  All points given by the adversary in this proof are given on a circle and all chords are non-crossing. Since all points are given on a circle, this also means that when two points are matched, a chord is created.
  This effectively prevents points from the region on one side of the chord to be matched with points in the region on the other side in the future processing
  by the online algorithm.
  For a given algorithm, \ALG,
  the adversary gives $2k$ points of weight~$1$.
  \ALG matches some pairs, creating $s$ chords.
  If $s < \frac{k}{3}$, the sequence ends, and we have a ratio of
  $\frac{2s}{2k}<\frac13$.
  Otherwise, we consider $s$ of the $s+1$ regions \ALG has divided
  the circle into separately.
  We ignore one region, so that we can associate exactly one
  chord with each of the remaining $s$ regions.
  No further points will be given in that region
  and any matching by \ALG of points from the first $2k$ points
  are counted in other regions,
  so we will not undercount \ALG's profit by ignoring that region.

  Let $S_4$ denote the regions with at least four unmatched points,
  and let $S_0$, $S_1$, $S_2$, and $S_3$ denote the regions with
  0, 1, 2, or 3 points, respectively.
  We treat a number of cases below. When we have treated a case, we
  list the profit of both \ALG and \OPT in parenthesis.  For a given
  region, we count the two points that are matched by the chord
  associated with the region.  Thus, both \ALG and \OPT always get a
  profit of at least two in any region (except the one we ignore).
  \OPT can of course process the points offline,
  matching all of the $2n$ points. However, for the
  proof, we count and compare the weighted points for both \ALG and
  \OPT in each region separately (sometimes two regions together when creating new regions).

  For the $S_0$-regions, the adversary gives one point of weight~$U$
  that \ALG cannot match (\ALG: $2$, \OPT: $U+2$).

  For the $S_1$-regions, the adversary gives one point of weight~$U$.
  If \ALG does not
  match the now two points, no further points are given to that region
  (\ALG: $2$, \OPT: $U+3)$.
  Otherwise, two points of weight~$U$ are given on either side of the line
  representing the latest match, and \ALG cannot match these
  (\ALG: $U+3$, \OPT: $3U+3$).

  For the $S_2$-regions, the adversary gives one more point of
  weight~$1$.  If \ALG does not match this point to any of the two
  others, we move to the $S_3$ case. Otherwise, \ALG has created an
  empty region and the adversary gives a point of weight~$U$ in that
  region that \ALG cannot match (\ALG: $4$, \OPT: $U+5$).

  For the $S_3$-regions, the adversary gives one more point of
  weight~$1$.  If \ALG does not match that point, no further points
  are given, and we are in the $S_4$ case.  Otherwise, if \ALG creates
  a region with no points, the adversary gives a point of weight~$U$
  in that region that \ALG cannot match (\ALG: $4$, \OPT: $U+6$).
  Otherwise, it has created two regions with one unmatched point each.
  We treat those two independently as in the $S_1$ case.

For the $S_4$-regions, no further points are given (\ALG: $2$, \OPT:
  at least~$6$).

  Having treated the cases, we observe that for large~$U$, the largest
  ratio between \ALG and \OPT is
  $\frac{U+3}{3U+3}= \frac13+\frac{2}{3U+3}$.
\end{proof}

\section{Deterministic Algorithms for Restricted \WNM}

In this section, we move beyond the two-weight case and consider
algorithmic decisions based on a classification schemes, where all
points with a weight falling within a given class are treated the same.

\subsection{Point Classification.}
\label{ssec:point_class}
In both lower and upper bound arguments, we use a point classification, based on parameters $k\in \mathbb{N}$ and $U\in \mathbb{R}$, explained here. 
Let 
    $k= \lceil \sqrt{\log U} \rceil$ 
and define values of $a_0, a_1, \ldots, a_k$ so that
    $$a_0 = 1, a_k = U, r=a_1/a_0=a_2/a_1=\ldots=a_k/a_{k-1},$$ 
which implies 
that 
    $r = U^{1/k}$ 
and 
    $a_i = r^i$. 
For a given value $w\in[1,U]$, define $\sff{w}$ to be the largest $a_i$ such that $a_i \leq w$. 
In what follows, a point with weight $w$ is said to have type $i$ if $\sff{w} = a_i$.
Thus, there are $k+1$ distinct types, with type $k$ containing only the value
$U$.
The type of a line segment between two matched points $x$ and $y$ is defined by the type of the end-point with larger weight, that is, $\overline{xy}$ has type $i$ if one of its endpoints has type $i$ and the other endpoint has type at most $i$.

\subsection{Negative Result}

\begin{theorem}\label{th:mainLo}
For the Restricted \WNM problem,
the asymptotic competitive ratio of any deterministic online algorithm is $O{\left(2^{-\sqrt{\log U}}\right)}$.
\end{theorem}

\begin{proof}
Let $\ALG$ be any online deterministic algorithm. 
We use an adversarial argument. The adversary gives all points on a circle, $C$, so any match the algorithm makes creates a chord in the circle, dividing a previous region into two. At any point in time, the adversary gives a point in an \emph{active region} of $C$, which is formed by one or two \emph{arcs}, the segments of the circle bounded by two consecutive points, in the boundary of $C$. Initially, the entire circle forms the active region. The adversary's strategy is to maintain a mapping from unmatched points to matched points to ensure the ratio between the total weight of matched points and unmatched points is $O{\left(2^{-\sqrt{\log U}}\right)}$. Note that this implies the ratio between the total weight of matched points and all points is also $O{\left(2^{-\sqrt{\log U}}\right)}$. 

The adversary starts by fixing an arbitrarily large number, $m$ (this is required to guarantee that our bound is asymptotic). The adversary gives points of weight $1$ in arbitrary positions on the circle until either the algorithm matches $m$ pairs of points or it reaches $m2^{k}$ points on the circle. In the latter case, the competitive ratio is at most 
    $O(2^{-k}) = O{\left(2^{-\sqrt{\log U}}\right)}$.

Therefore, we may assume that $\ALG$ eventually matches $m$ pairs of points, creating non-intersecting chords, and $m+1$ regions. Now, make each matched pair responsible for a distinct region created, though with the first matched pair being responsible for two regions, initially the first two regions. Suppose a new chord $\overline{xy}$ divides region $R$ into two. Let $\{ x_R, y_R\}$ be the responsible pair for $R$, $R_1$ the side of $R$ that has $\overline{x_Ry_R}$ on its boundary, and $R_2$ the other side. Leave $\{ x_R,y_R\}$ responsible for $R_1$ and make $\{ x,y\}$ responsible for $R_2$. This ensures that each matched pair, except the first, is responsible for exactly one region, and the first is responsible for two.

For each region $R$, the adversary makes $R$ the active region, runs the following procedure and continues with the next region until it has covered all regions. Let $\{ x_R,y_R\}$ be the responsible pair of points for $R$. Consider the following two cases, depending on the number of unmatched points in $R$:

\noindent\textbf{Case 1.} If the number of unmatched points in $R$ is at least $2^k-1$, the adversary does not give any point in $R$ and continues to the next region. In this case, we map the unmatched points in $R$ to the matched pair $\{ x_R,y_R\}$. Note that at least $2^k-1$ points of weight 1 are mapped to a segment of total weight 2. The ratio of the weight of matched points to the unmatched points will be at most
        $2/(2^k-1) \in O{\left(2^{-\sqrt{\log U}}\right)}$.

\noindent\textbf{Case 2.} If the number of unmatched points in $R$ is less than $2^k-1$, the adversary plans to give a sequence of points, $P = \langle p_1, p_2, \ldots, p_k\rangle$, with weights $a_1,a_2,\ldots,a_k$, respectively, in ascending order of their weights, in the following manner (see Fig.~\ref{fig:general-lower-bound}). The point $p_1$ of weight $a_1$ appears in an arbitrary position in $R$ (on the circle). Upon the arrival of a point $p_i$ with weight $a_i$, $i \in \{1,\ldots, k\}$, either $\ALG$ matches it with a point of weight 1 or leaves it unmatched. In the latter case, the adversary does not give more points in $R$ and continues with the next region. 

In the former case, when $\ALG$ matches a point $p_i$ of weight $a_i$ with a point $q$ of weight $1$, make a side of $\overline{p_iq}$ that contains at most half of the unmatched points, the active region. The adversary continues giving the remaining points of $P$ in the active region.
 Thus the unmatched points on the opposite side of $\overline{p_iq}$ stay unmatched, since $\overline{p_iq}$ is between the new points and these unmatched points.
 
Therefore, after matching $p_i$ and $q$, the number of unmatched points of weight $1$ that can match with future points in $P$ decreases by a factor of at least $2$. Let $p_j$ be the first point in $P$ that the algorithm leaves unmatched. Given that the adversary can give up to $k$ points, and there are initially fewer than $2^{k}-1$ unmatched points in $R$, there exists such $p_j$ of weight $a_j$. At this point, the adversary ends the procedure for $R$ and continues with the next region.

The total weight of points in matched pairs in $R$ before the arrival of $p_j$ is:
\begin{align*}
    M  &= \underbrace{2}_{\text{for } (x_R,y_R)}+ \underbrace{j-1}_{\text{endpoints of weight } 1}+\underbrace{a_1+a_2+\ldots+a_{j-1}} _{\text{endpoints with weight } a_i }
    \qquad \\ &\le \qquad \frac{r^j-1}{r-1} + j
\end{align*}
Given that the unmatched point $p_j$ is of weight $a_j$, the ratio between the weight of matched points and unmatched points is at most
    $M/r^j \in O(1/r) = O{\left(2^{-\sqrt{\log U}}\right)}$.

Given that each matched pair is responsible for at least one region, the above procedure creates a mapping of matched points to unmatched points with a weight ratio of $O(2^{-\sqrt{\log U}})$ in all cases, as desired.
\end{proof}

\begin{figure}[htb]
\centering
\includegraphics[scale=0.7]{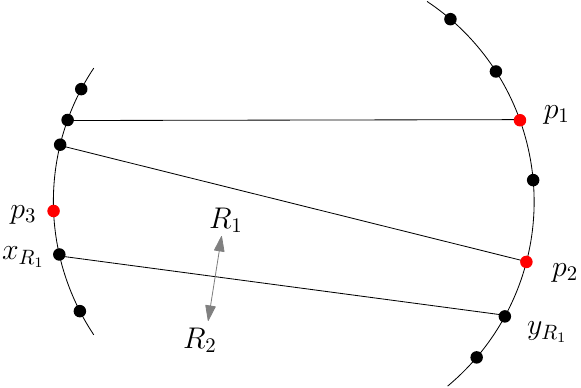}
    \caption{An illustration of the adversary's strategy for Case 2 with $k=3$. The two arcs form the active region. The unnamed points have weight 1. Suppose in the first phase $\ALG$ matched $x_{R_1}$ and $y_{R_1}$, which became responsible for region $R_1$. Note that the number of unmatched points (of weight 1) in $R_1$ is 6, which is less than $2^k-1 =7$. Thus, in the second phase, the adversary plans to give points $p_1, p_2,p_3$ of weights $a_1, a_2, a_3$ into $R_1$. Suppose $\ALG$ matches $p_1$ of weight $a_1$; then the adversary gives $p_2$ and $p_3$ below the line segment between the matched pair (there are fewer unmatched points there). Similarly, after the point $p_2$ of weight $a_2$ is matched, the adversary gives $p_3$ to the side of the resulting segment with no unmatched points. This ensures that some point of weight $a_i$ (here $a_3$) stays unmatched and is mapped to the matched pairs.}
    \label{fig:general-lower-bound}
\end{figure}
 
\subsection{Positive Result: The Wait-and-Match Algorithm}

We propose an algorithm called ``Wait-and-Match'' (\WAM). 
Assume the points appear in a bounding box $B$. 
Throughout its execution, \WAM maintains a \emph{convex partitioning} of $B$. Initially, there is only one region formed by the entire $B$. As we will describe, the algorithm matches two points only if they appear in the same convex region. 
Whenever two points in a convex region $R$ are matched, the line segment between them is extended until it hits the boundary of $R$, which results in partitioning $R$ into two smaller convex regions. We use the same point classification as defined in Section~\ref{ssec:point_class}.

Suppose a new point $p$ appears, and let $R$ denote the convex region of $p$.
In deciding which point to match $p$ to (if any), the algorithm considers all unmatched points in $R$ in non-increasing order of weights. Let $q$ be the next point being considered, and let $i$ be the maximum of the type of $p$ and the type of $q$.
The algorithm matches $p$ with $q$ if there are at least $2^{k-i}-1$ unmatched points on each side of $\overline{pq}$.  
If no suitable $q$ exists, $p$ is left unmatched.

\medskip\noindent {\bf Example:}
Suppose $k=2$. Then $a_0=1, a_1 = \sqrt{U},$ and $a_2 = U$. Let $p$ be a point with weight $1$. Upon the arrival of a point $p$, the algorithm matches $p$ with any point $q$ of weight $U$, if there are at least $2^{2-2} -1 = 0$ points on each side of $\overline{pq}$. That is, if there is an unmatched point of weight $U$ in the region, the algorithm would match $p$ to it unconditionally.
Similarly, if there are no unmatched points of weight $U$ in the region, the algorithm tries to match $p$ with any point $q$ of weight in $[a_1=\sqrt{U}, a_2=U)$, provided there is at least $2^{2-1}-1 = 1$ point on each side of $\overline{pq}$. Finally, if previous scenarios do not occur, the algorithm tries to match $p$ with any point $q$ of weight in $[a_0=1, a_1=\sqrt{U})$ provided there are at least $2^{2-0}-1=3$ unmatched points on each side of $\overline{pq}$. This will happen if there were at least $7$ unmatched points in the region.
\hspace*{\fill}\raisebox{-1pt}{\boldmath$\Diamond$}

\medskip
To analyze the algorithm, we 
match each unmatched point to a matched pair. For the sake of the analysis, we introduce two ``imaginary'' points $(-\infty,0)$ and $(-\infty,1)$ of weight $U$ and treat them as if they were matched before the input sequence is revealed. Suppose a new point, $p$, arrives in a region $R$. We map $p$ to the most recent segment that forms a boundary of the region $R$.  See Figure~\ref{fig:upper-bound} for an illustration of this mapping.

\begin{figure}
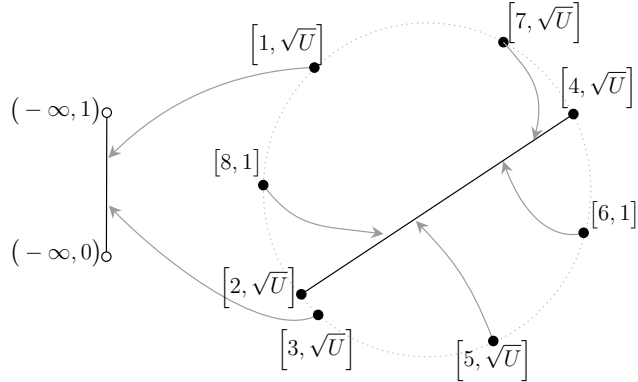

    \centering
    \scalebox{.6}{\figMapping}
    \caption{An illustration of the mapping used to analyze \WAM. In this example, we have $k=2$ and $8$ points with weights in $\{1,\sqrt{U}, U\}$. Here, $[t, w]$ indicates the $t^{\text{th}}$ point in the input sequence having weight $w$. Note that points $1$ and $3$ are mapped to the segment corresponding to the imaginary points $(-\infty,0)$ and $(-\infty,1)$ of weight $U$.}
    \label{fig:upper-bound}
\end{figure}

\begin{lemma}\label{lem:lowerIndx}
 Every point of type $i$ is mapped to a segment of type at least~$i$.
\end{lemma}

\begin{proof}
 For the sake of contradiction, suppose a point $p$ with type $i$ arrives in the region $R$ and gets mapped to $\overline{p_R,q_R}$  of type at most $i-1$.
 
 Without loss of generality, assume $p_R$ arrived after $q_R$.
By the definition of the algorithm, at the time $p_R$ appeared, there were at least $2^{k-i+1}-1$ unmatched points in $R$ (otherwise, $p_R$ would not have been matched with $q_R$). These unmatched points are still unmatched when $p$ appeared (otherwise, $R$ should have been partitioned, and $p$ should have been mapped to some other segment). 
 Thus, when $p$ appeared, the algorithm could match it with the point that bisects these unmatched points, and there would be at least $(2^{k-i+1}-2)/2 = 2^{k-i}-1$ points on each side of the resulting line segment. This contradicts the fact that the algorithm left $p$ unmatched.  
\end{proof}

\begin{lemma}\label{lem:upperIndx}
    Let $s$ be any line segment between two matched points. For any $i$, at most $2^{k-i+2}-2$ unmatched points of type $i$ are mapped to $s$.
\end{lemma}

\begin{proof}
For the sake of contradiction, assume at least $2^{k-i+2}-1$ points of type $i$ are mapped to $s$. Then, there must be at least $\lceil (2^{k-i+2}-1)/2\rceil = 2^{k-i+1}$ points of type $i$ in a convex region $R$ formed by extending $s$. At the time the last of these points, say $p$, arrives, it could be matched to the point $q$ that bisects the other points; there will be at least $(2^{k-i+1}-2)/2 = 2^{k-i}-1 $ points on each side of $\overline{pq}$. Since $\overline{pq}$ is of type $i$, the algorithm must have matched $p$ with $q$, which contradicts the fact that $p$ and $q$ are unmatched and mapped to $s$.  
\end{proof}

\begin{lemma}\label{lem:mapping}
    For $U \geq 16$, the total weight of unmatched points mapped to a segment of type $j$ is at most $a_{j+1} 2^{k-j+3}$.
\end{lemma}

\begin{proof} 
    Note that a point of type $i$ has weight at most $a_{i+1} = r^{i+1}$. Hence, by Lemma~\ref{lem:lowerIndx} and Lemma~\ref{lem:upperIndx}, the total weight of unmatched points mapped to a segment of type $j$ is at most
    \[
        \sum_{i=0}^j r^{i+1} 2^{k-i+2}
        = 2^{k+2} r \sum_{i = 0}^j \left(\frac{r}{2}\right)^i 
        = 2^{k+2} r \frac{(r/2)^{j+1} - 1}{r/2-1}  
        \le a_{j+1} 2^{k-j+3}. 
    \]
    Here, we have assumed that $r \ge 4$, which holds for $U\geq 16$, since
    $16^{\frac{1}{\lceil\sqrt{\log 16}\rceil}}=(2^4)^{\frac{1}{2}}=4$.
\end{proof}

\begin{theorem}\label{th:mainUp}
For sufficiently large $U$,
the competitive ratio of the deterministic online algorithm \WAM for the Restricted \WNM problem is 
    $$\Omega\left(2^{-2\sqrt{\log U}}\right).$$
\end{theorem}

\begin{proof}
For every matched pair $\overline{pq}$ by \WAM consider the set of points formed by $p$, $q$, and the unmatched points mapped to them. By Lemma~\ref{lem:mapping}, if $\overline{pq}$ has type $j$, the ratio of the weight of the matched pair over all the points in this set is at least $\frac{a_j}{2a_j+a_{j+1}2^{k-j+3}}\geq \frac{1}{r2^{k+4}}$. 

With the values from Section~\ref{ssec:point_class} of
$k = \lceil \sqrt{\log U} \rceil$ and $r = U^{1/k}$, it follows that
$r = 2^{(\log U) / k} \le 2^k$.
Thus, since the algorithm \WAM guarantees that every unmatched point is mapped to some matched pair, the competitive ratio of \WAM is at least $2^{-(2k+4)}$.
\end{proof}

\section{Randomized Algorithms} \label{sec:randomized}
Having treated the deterministic case, we investigate how much better the expected competitive ratios can become when using randomization.

\subsection{Negative Result}

By Yao's minimax principle~\cite{Y77p}, to show that the competitive ratio of randomized algorithms is at most $\gamma$, it suffices to present a randomized input on which every deterministic algorithm has a competitive ratio at most $\gamma$. We will create a randomized unweighted input similar to what Lavasani and Pankratov~\cite{lavasani2023} used for the advice model, as follows.

\begin{figure}
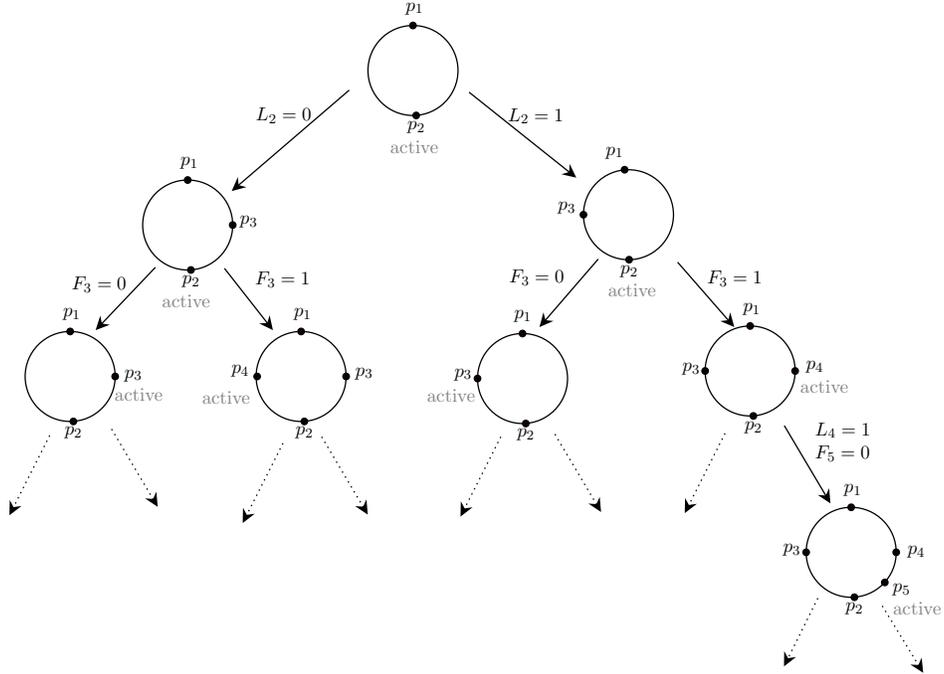

\centering
\scalebox{.65}{\figTreeLowerBoundUpdated}
    \caption{An illustration of the adversary's construction of the randomized input in the first few steps, guided by the random variables $L_i$ and $F_i$.}
    \label{fig:thm8_randomizedinput}
\end{figure}

The input points all lie on a circle; see an illustration in Figure~\ref{fig:thm8_randomizedinput} for the first few steps of the adversary's construction.  For a point $p$ on the circle, let the left and the right arcs of $p$ be the clockwise and counter-clockwise arcs that are bounded by $p$. Initially, an adversary gives $p_1$ and $p_2$ on two arbitrary antipodals of the circle, creating two arcs. The adversary makes $p_2$ the current \emph{active} point. For the rest of the points, the adversary uses two sequences $L_1,\ldots, L_{2n}$ and $F_1,\ldots,F_{2n}$ of independent identically distributed Bernoulli random indicator variables with parameter $1/2$ to guide the decision. If $p_i$ is the current active point, $L_i$ determines the position of $p_{i+1}$. If $L_i$ is $1$, the adversary gives $p_{i+1}$ in the middle of the left arc of $p_i$, the arc formed by $p_i$ and the closest point on its left. If $L_i$ is $0$, the adversary gives it in the middle of the right arc of $p_i$. Given $p_i$ is an active point, if $F_{i+1}$ is $1$, the point $p_{i+1}$ becomes, what we will call, a \emph{skipped} point.  In this case, the adversary makes $p_{i+2}$ the next active point and gives it in the middle of the other arc of $p_i$ (for example, if $p_{i+1}$ is on the left arc of $p_i$, then $p_{i+2}$ is given on the right arc of $p_i$). If $F_{i+1}$ is $0$, the adversary makes $p_{i+1}$ the new active point. The adversary continues the procedure with the new active point until $2n$ points have been created.

\begin{theorem}
\label{thm:rand_lb}
    No randomized online algorithm can achieve a competitive ratio better than $16/17$ in expectation.
\end{theorem}

\begin{proof}
As mentioned at the beginning of this subsection, by Yao's minimax principle~\cite{Y77p}, it suffices to study the behavior of deterministic algorithms on the randomized input constructed as described above, illustrated in Figure \ref{fig:thm8_randomizedinput}. Now we fix a deterministic algorithm, $\ALG$. Segments of points matched by $\ALG$ divide the circle into convex regions.  If an unmatched point is in a region where no new points arrive, it cannot be matched anymore and we call it an \emph{isolated} point. Given that \ALG matches $p_i$ upon its arrival, let $X_i$ be the indicator random variable that $p_i$ is an active point and matching it causes at least one point (either be a previous point or the point $p_{i+1}$) to become isolated. Our aim is to prove a lower bound on  the number of isolated points, or equivalently, an upper bound on the number of segments.

We express $X_i$ in terms of other random variables. Let $A_i$ be the indicator random variable that $p_i$ becomes an active point. For $i\geq 3$, $p_i$ is a skipped point if and only if $p_{i-1}$ is an active point and $F_i$ is $1$. Thus, 
\begin{equation}	\label{eq:bound_Ai}
	A_i = 1-A_{i-1}F_i \ge 1- F_i. 
\end{equation}
Suppose $p_i$ becomes an active point that $\ALG$ matches upon its arrival in a convex region $R$, splitting $R$ into $R_L$ and $R_R$, containing the left and right arcs of $p_i$, respectively. We say $R_L$ (or $R_R$) is empty if it does not contain any unmatched point. We have the following four cases:

\begin{itemize}
    \item $R_L$ and $R_R$ are both empty: $p_{i+1}$ becomes isolated if it is a skipped point. Hence, $X_i \ge A_i F_{i+1}$;

    \item $R_L$ and $R_R$ are both non-empty: suppose $p_{i+1}$ lies in $R_L$ and becomes an active point, then unmatched points in $R_R$ become isolated. The case when $p_{i+1}$ lies in $R_R$ is similar. Hence, $X_i \ge A_i (1-F_{i+1})$;

    \item $R_L$ is empty and $R_R$ is non-empty: If $p_{i+1}$ arrives on the left arc of $p_i$ ($L_i=1$), then at least one isolated point will be guaranteed. Indeed, if $p_{i+1}$ is a skipped point, it becomes isolated because $R_L$ is empty, and if it is the new active point, unmatched points in $R_R$ become isolated. Hence, $X_i \ge A_i L_i$; 

    \item $R_L$ is non-empty and $R_R$ is empty: Symmetrically, $X_i \ge A_i (1-L_i)$. 
\end{itemize}
If $i = 2n$, there is no $p_{i+1}$, and matching $p_i$ makes points isolated if $R_L$ or $R_R$ are not empty. Since we are interested in the asymptotic competitive ratio, we ignore this boundary case. To sum up, we get

\begin{equation}    \label{eq:def_Xi}
    X_i \ge 
    \begin{cases}
        A_iF_{i+1}, & \text{ if $R_L$ and $R_R$ are empty } \\
        A_iL_{i}, & \text{ if $R_L$ is empty and $R_R$ is non-empty  } \\
        A_i(1-L_i), & \text{ if $R_L$ is non-empty and $R_R$ is empty   }\\
        A_i(1-F_{i+1}), & \text{ if $R_L$ and $R_R$ are non-empty }
    \end{cases}
\end{equation}
Define $Y_i$ similarly to $X_i$ in \eqref{eq:def_Xi}, but replace $A_i$ by $1-F_i$, i.e., 
\begin{equation}    \label{eq:def_Yi}
    Y_i = 
    \begin{cases}
        (1-F_i) F_{i+1}, & \text{ if $R_L$ and $R_R$ are empty } \\
        (1-F_i) L_{i}, & \text{ if $R_L$ is empty and $R_R$ is non-empty  } \\
        (1-F_i) (1-L_i), & \text{ if $R_L$ is non-empty and $R_R$ is empty   }\\
        (1-F_i) (1-F_{i+1}), & \text{ if $R_L$ and $R_R$ are non-empty }
    \end{cases}
\end{equation}
Then, \eqref{eq:bound_Ai} implies $Y_i \le X_i$.

Now, let $M$ be a random variable that denotes the size of the matching made by $\ALG$, and for each $1\leq i \leq M$, let $T_i$ be the step number in which $\ALG$ makes the $i^{\text{th}}$ match. Note that $T_i$ and $M$ are completely determined by the sequences $L_i$ and $F_i$ (for $1 \le i \le 2n$) and the algorithm $\ALG$.
Thus, the algorithm is guaranteed to have at least $\sum_{i = 1}^M X_{T_i}$ unmatched points. Since each segment uses two points, we have 
	$\sum_{i = 1}^M (X_{T_i} + 2) \le 2n$.
Hence,  by $Y_i \le X_i$, we have, 
\begin{equation}    \label{eq:half}
    2n \ge \sum_{i = 1}^M (Y_{T_i} + 2) 
    \ge \sum_{i = 1}^{\lfloor \frac{M}{2}\rfloor} (Y_{T_{2i-1}} + Y_{T_{2i}} + 4)
    \ge \sum_{i = 1}^{\lfloor \frac{M}{2}\rfloor} ( Y_{T_{2i}} + 4).
\end{equation}
Define another auxiliary random variable sequence $Z_1, Z_2,\ldots$ as follows:
\[
    Z_i = 
    \begin{cases}
        Y_{T_{2i}} + 4, &\quad 1 \le i \le \lfloor \frac{M}{2}\rfloor, \\
        Y' + 4, &\quad i > \lfloor \frac{M}{2}\rfloor,
    \end{cases}
\]
where $Y'$ is an i.i.d.\ Bernoulli random variable with parameter $1/4$. By the definition of $Y_i$, because $Y_{T_{2i}}$ and $Y_{T_{2(i+1)}}$ depend on disjoint sets of variables,  $Y_{T_{2i}}$ are i.i.d.\ Bernoulli random variables with parameter $1/4$. Hence, the sequence $Z_i$ for all $i\ge 1$ is i.i.d.\ with expectation $17/4$. Let the random variable $M'$ be the maximum integer such that $\sum_{i=1}^{M'} Z_i \le 2n$. Note that by definition of $Z_i$ and \eqref{eq:half}, we have $\sum_{i=1}^{\lfloor \frac{M}{2}\rfloor} Z_i \le 2n$, and, hence, $M' \ge \lfloor \frac{M}{2}\rfloor$. Since the $Z_i$'s are i.i.d.\ with expectation $17/4$, by the renewal theorem, we have $\Ex[M'] = 2n/(17/4) = 8n/17$. Hence, 
    $8n/17 \ge \Ex[\lfloor \frac{M}{2}\rfloor] \ge \Ex[M]/2 - 1$,
this implies
    $\Ex[M] \le 16n/17 + 2$. 
As a result, in expectation, the deterministic algorithm $\ALG$ can match $2\Ex[M] \le 2(16n/17 + 2)$ points. Therefore, the expected competitive ratio is 
$\frac{2\Ex[M]}{2n}$, the infimum of which, for $n\to \infty$, is~$16/17$.
\end{proof}

\subsection{Positive Result: Tree-Guided-Matching Algorithm}

We propose a randomized algorithm called ``Tree-Guided-Matching'' (\TGM) that has the following uniform guarantee, regardless of the weights of the points: each point appears in a matching with probability at least~$1/3$. The algorithm \TGM uses a binary tree to guide its matching decisions. The structure of the binary tree is determined by the relative positions of the input points, with each node of the tree corresponding to such a point. Intuitively, the binary tree, as it grows, gives an online refinement of the partition of the plane into convex regions, such that for each region, there is some online point \emph{responsible} for it; see Figure~\ref{fig:TGM}.
We emphasize that the regions and the tree that is created to represent them do not indicate any matching. Their purpose is to ensure that if matchings are made, then they are non-crossing, and it also indicates which points we will try to match to with some probability.

\begin{figure}[ht!]
\centering
\includegraphics[scale=0.15]{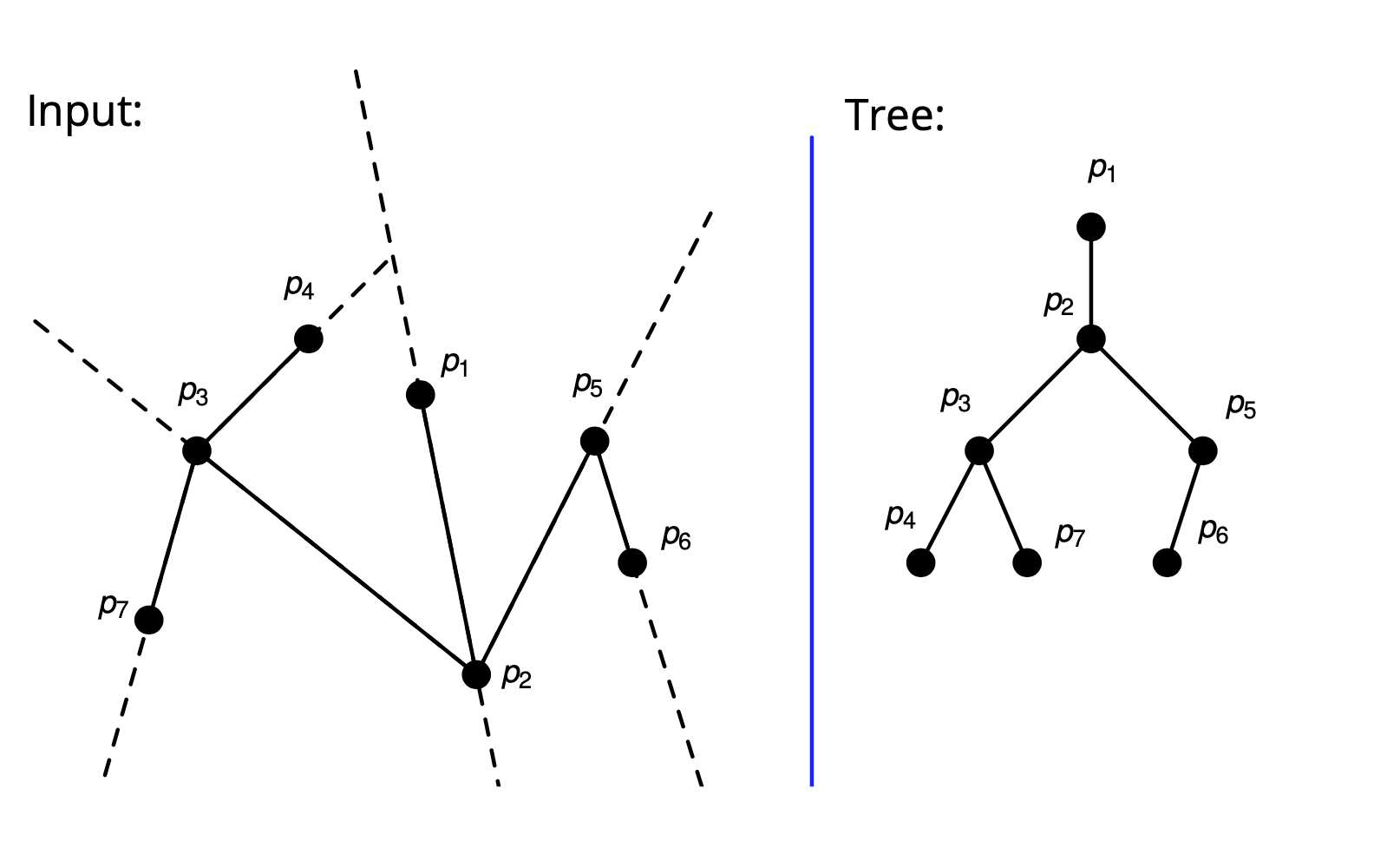}
    \caption{On the left is the input and how \TGM divides and partitions the plane, and on the right is the tree it creates from the input. Based on this tree, it matches nodes with their parents randomly, such that every node upon its arrival gets matched with a probability of $1/3$.}
    \label{fig:TGM}
\end{figure}

Initially, set $p_1$ as the root of the tree and $p_2$ as the child of $p_1$. Abusing notation slightly, we also use $\overline{pq}$ to denote the straight line determined by points $p$ and $q$. Let $R_1$ and $R_2$ denote the two regions corresponding to the half-spaces created by $\overline{p_1p_2}$.
Again, we emphasize that this defines the start of the tree and the regions
and does not indicate that $p_1$ and $p_2$ are matched.
The decision as to whether or not that will happen id discussed below.
Let $p_2$ be responsible for both $R_1$ and $R_2$. In general, when $p_i$ arrives into a region $R$ for which $p_j$ is responsible (of course, $j < i$), make $p_i$ a child of $p_j$ in the binary tree. The line $\overline{p_i p_j}$ divides the region $R$ into two sub-regions $R'$ and $R''$,  let $p_i$ be responsible for both of them, and at this point the responsibility of $p_j$ on $R$ is lost as region $R$ has been refined to $R'$ and $R''$. Note that this implies that every node of the tree has at most two children. 
Next, we describe how \TGM chooses to match points. At the beginning, \TGM matches $p_2$ with $p_1$ with probability $1/3$. After that, upon the arrival of $p_i$, let $p_j$ be its parent in the tree. If $p_j$ is unmatched and $p_i$ is its first child,  match $p_i$ to $p_j$ with probability $1/2$. If $p_j$ is unmatched and $p_i$ is its second child,  match $p_i$ to $p_j$ deterministically. Note that \TGM only tries to match an online point with its parent in the tree.

\begin{theorem}
\label{thm:rand_ub}
    Every point, regardless of its weight, is matched by the randomized algorithm \TGM with probability at least $1/3$. Hence, \TGM 
    achieves a strict competitive ratio of at least $1/3$. 
\end{theorem}

\begin{proof}
    Note that since \TGM only matches a child to its parent in the binary tree, the matching is non-crossing. Indeed, by our construction of the tree, recalling that points are in general position, every child is a point \emph{inside} a convex region for which its parent is responsible, and its parent lies on the boundary of that region. Hence, the line segment formed by a match does not cross any existing line segment. 

    Next, we show the claimed performance of \TGM. 
    By the definition of \TGM, $p_1$ is matched (with $p_2$) with probability $1/3$. 
    We will show that every $p_i$, $i\ge 2$, \emph{upon its arrival} gets matched to its parent with probability exactly $1/3$, which implies the claim. To see this, we proceed inductively. The base case is true for $p_2$. Let $p$ be the currently arrived point, and let $q$ be its parent. 
    We consider two cases. 
    \begin{itemize}
        \item If $p$ is the first child of $q$, then by the induction hypothesis, $q$ is unmatched at this moment with probability $2/3$. Hence, according to \TGM, $p$ is matched (with $q$) with probability $(2/3) \cdot (1/2) = 1/3$.

        \item If $p$ is the second child of $q$, then $q$ is unmatched at this moment with probability $1-1/3 -1/3 = 1/3$. By \TGM, $p$ is matched (with $q$) with probability $(1/3) \cdot 1 = 1/3$. 
    \end{itemize}
\end{proof}

\section{Revocable Acceptances}
\label{sec:revoke}

In this section, we consider the revocable setting. When a new point $p$ arrives, an algorithm has the option of removing one of the existing edges from the matching prior to processing $p$. The decision to remove an existing edge is irrevocable. The benefit of making this decision is that the end-points of the removed edge, along with possible points on the other side of the edge (though our positive result does not use this possibility), become available candidates to be matched with $p$, provided the non-crossing constraint is maintained.

\subsection{Negative Result}

Bose et al.~\cite{bose2020non} showed that a deterministic greedy algorithm without revoking can achieve a competitive ratio of $2/3$ in the unweighted version. In this section, we prove that, even allowing revoking, this is optimal.

\begin{theorem}
\label{thm:revoke_negative}
Even in the unweighted version,
    no deterministic algorithm with revoking can achieve a competitive ratio better than $2/3$.
\end{theorem}

\begin{proof}
    Fix a deterministic algorithm \ALG, an arbitrarily large $n$, and a circle in the plane. The adversary gives at least $2n$ points, all of weight~$1$, on the circle, one by one, and let \ALG match them into pairs. Thus, all matches are chords on the circle. We maintain the invariants that there is always one active region of the circle, and that for each matched pair, there is always exactly one unique unmatched point associated with it. Such a point is not in the active region and in its region, there are no other points that are associated with a matched pair. Additionally, the chord of matched pair it is associated with borders its region. There may be additional unmatched points that are \emph{not} associated with a matched pair.

    Initially, the entire circle is the active region. A phase consists of the adversary presenting points on the circle, in the active region, until \ALG either matches a pair or revokes a matching (or the total of $2n$ points have been given).

We first consider the case, where there is no revoking.
Thus, the current point, $p$, is simply matched to a point $q$ on the circle. The chord $\overline{pq}$ divides the active region into two sub-regions, $R_1$ and $R_2$.
Without loss of generality, assume that $R_1$ contains at least as many unmatched points as $R_2$.
In case neither region has any points, we give a point $p'$ to $R_1$.
Now there is some unmatched point in $R_1$, and one of those is associated with the matched pair,
$R_2$ becomes the active region, the phase ends, and the invariants clearly hold.
    
Now assume that \ALG revokes a matching.
Removing the one match removes a chord of the circle, joining two regions into a new convex region.
Note that one of the regions was the active one, so in that region there are no unmatched points associated with a matched pair.
The chord only borders two regions, so the other region must
contain the associated point.
Thus, after revoking, the joint region has no points associated with a matched pair.

If $p$ is not matched, this joint region becomes the active region, and the phase ends.

If $p$ is matched, it may be to one of the points from the revoked match.
However, at least one of the points, $q$, from the revoked match is now unmatched since $p$ is only matched to one point.
Since $q$ was just part of a matching, it cannot have been an unmatched point
associated with some matching, so we can associate $q$ with the new matching
just created.
The sub-region created by the match that does not contain $q$ becomes the active region, and the current phase ends.

    Inductively, the invariants hold after each phase, and the unmatched point associated with each matched pair ensures that no more than $2/3$ of the points are matched. Although the number of points may be odd, this gives an asymptotic lower bound of $2/3$ on the competitive ratio.
\end{proof}

\subsection{Positive Result: Big-Improvement-Match}
We present a deterministic algorithm with revoking, called ``Big-Improvement-Match'' (\BIM). This algorithm has a strict competitive ratio of approximately $\bimrho$, even when weights of points are unrestricted. This shows that while revoking does not improve the competitive ratio in the unweighted version, it provides us with an algorithm with a constant competitive ratio, which is unattainable for a deterministic algorithm without revoking.

In the following, we abuse notation slightly, using $\overline{pq}$ to denote the straight line determined by two points $p$ and $q$.
\BIM maintains a partitioning of the Euclidean space into regions. Each region in the partition is assigned an edge from the current matching to be responsible for that region. Each edge can be responsible for up to two regions.
\BIM starts out by matching the first two points, $p_1$ and $p_2$ regardless of their weights, dividing the plane into two half-planes by $\overline{p_1p_2}$.
 \BIM then assigns $\overline{p_1p_2}$ to be responsible for the two half-plane regions. Next, consider a new point $p_i$ (for $i \ge 3$) that arrives in an existing region, $R$. Suppose that $\overline{p_j p_{j'}}$ is the responsible edge for $R$. If there is at least one unmatched point in $R$, \BIM matches $p_i$ with an unmatched point $p_k$ in $R$ with the highest weight. Then $\overline{p_j p_{j'}}$ is no longer responsible for $R$, and the region $R$ is divided into two new regions by $\overline{p_ip_k}$. The responsibility for both new regions is assigned to $\overline{p_i p_k}$. If $p_i$ is the only point in $R$, then \BIM considers revoking the matching $(p_j, p_{j'})$ as follows. Without loss of generality, assume that $w(p_j)\leq w(p_{j'})$.
 We introduce a parameter, $r$, that is going to be chosen later so as to optimize the competitive ratio.
 If $w(p_i)<r w(p_{j'})$,
then \BIM leaves $p_i$ unmatched. Otherwise, \BIM removes the matching $(p_j,p_{j'})$ and matches $p_i$ with $p_{j'}$. If $R$ is the only region that $\overline{p_jp_{j'}}$ was responsible for when $p_i$ arrived, then $R$ is divided into two regions by $\overline{p_ip_{j'}}$, and $\overline{p_i p_{j'}}$ becomes responsible for the two new regions. (The regions on the other side of $\overline{p_j p_{j'}}$ from $p_i$ keep their boundaries, even though $(p_j,p_{j'})$ is no longer in the matching.) Otherwise $\overline{p_j p_{j'}}$ was responsible for $R'$ in addition to $R$ when $p_i$ arrived. In this case, after removal of the match $(p_j, p_{j'})$, regions $R$ and $R'$ are merged to give region $R''=R\cup R'$, and $R''$ is divided by $\overline{p_ip_{j'}}$ into two regions, and \BIM makes $\overline{p_i p_{j'}}$ responsible for both new regions.

\begin{proposition}
\label{prop}
  The following observations concerning \BIM hold:
\begin{enumerate}
  \item \label{obs:responsible} All responsible edges are defined by two currently matched points.
\item \label{obs:atmosttworegions} Each edge is responsible for at most two regions.
\item \label{obs:convex} All regions are convex.
  \item \label{obs:replaced} When a matched edge $(p_j,p_{j'})$ is replaced due to the arrival of a point $p_i$ in region $R$, then edge $(p_i,p_{j'})$ lies entirely in $R$.
\end{enumerate}
\end{proposition}
\begin{proof}
  (\ref{obs:responsible}) follows since an edge only becomes responsible when its endpoints become matched. When another edge becomes responsible for a region, the original edge is no longer responsible.
    (\ref{obs:atmosttworegions}) follows since the only two regions an edge is made responsible for are the two regions created when the endpoints of the edge were matched. When two points in one of the regions an edge is responsible for are matched, the edge is no longer responsible for that region, but will still be responsible for one region if it had been responsible for two up until that point.
  (\ref{obs:convex}) follows inductively, since separating two convex regions by a line segment creates two convex regions. In addition, when \BIM removes an edge, that edge was the last matching created in either of the two regions it was responsible for.
  (\ref{obs:replaced}) follows by (\ref{obs:convex}).
  \end{proof}

\begin{theorem}
\label{thm:revoke_ub}
    For the \WNM problem with revoking and arbitrary weights,
    \BIM with  $r \in (1,\sqrt{2}]$ has strict competitive ratio at least
    $\min\left\{ \frac{r^2-1}{r^3}, \frac{1}{1+2r} \right\}$.
\end{theorem}

\begin{proof}
For each region an edge is responsible for, we consider the total weight of unmatched points in that region. These points fall into two categories: those that were matched at some point during the execution, but due to revoking became unmatched, and those that were never matched during the entire execution of the algorithm.

  Consider any subsequence of all created edges, $\langle e_1,\ldots, e_k\rangle$, where $e_1$
  was created when a second unmatched point arrived in some region, and the possible remaining edges were created via revokings, i.e., $e_i$ caused $e_{i-1}$ to be revoked for $2\leq i \leq k$, and $e_k$ is in \BIM's final matching. Let $e_j=(p_{i_j},p_{i_{j+1}})$ and $w(p_{i_j})\leq w(p_{i_{j+1}})$, so $e_{j+1}= (p_{i_{j+1}},p_{i_{j+2}})$. Thus, for $3\leq j \leq k+1$, $p_{i_j}$ arrived after $p_{i_{j-1}}$. Every pair ever matched by \BIM is included in some such sequence of edges. The points, $p_{i_1},\ldots,p_{i_{k-1}}$ could be unmatched points in a region for which $e_k$ is responsible.

We bound the total weight of points that were at some point in a region that one of the $e_j$'s were responsible for. We divide this up and first consider the points that were matched during the execution but no longer are (the $p_{i_j}$'s) and consider points that were never matched afterwards.
  
  Let $\alpha = w(p_{i_k})$, so for every $2\leq j \leq k$, $w(p_{i_j})$ is at most $\alpha r^{-(k-j)}$ and $w(p_{i_1})\leq w(p_{i_2}) \leq \alpha r^{-(k-2)}$. Let $\beta = w(p_{i_{k+1}})$.
    The total weight of points in this sequence is
    \begin{align*}\sum_{j=1}^{k+1} w(p_{i_j}) &= w(p_{i_1}) + \sum_{j=2}^{k} w(p_{i_j}) + w(p_{i_{k+1}}) \\ &\leq 
\alpha r^{-(k-2)} + \alpha \left(\frac{r}{r-1}\right)(1-r^{-(k-1)}) + \beta\\
&= \alpha \left(r^{-(k-1)} \frac{r(r-2)}{r-1} + \frac{r}{r-1} \right) + \beta.
\end{align*}
    Now, we consider other points that were never matched, but were at
    some time in a region for which one of the $e_j$ was responsible.
    After $e_1$ is created and before $e_2$, a first point $q_1$ could arrive in one of the regions for which $e_1$ is responsible. Since $q_1$ is not matched to $p_{i_2}$, $w(q_1)<r w(p_{i_2})$. (Note that a second point arriving in that region will then be matched to $q_1$, dividing the region, and the sub-regions will not be considered part of the region for which $e_k$ eventually becomes responsible.) Now, suppose that another point, $q_2$, arrives between when $e_j$ and $e_{j+1}$ are created for some $2\leq j<k$, remaining unmatched in one of the regions for which $e_k$ is responsible. Then, neither $q_2$ nor $p_{i_{j+2}}$ is in the same region as $p_{i_{j-1}}$ or one of them would have been matched to $p_{i_{j-1}}$ (or $p_{i_{j-1}}$ was already matched and the region divided). 
    By Proposition~\ref{prop}.\ref{obs:atmosttworegions}, $e_i$ is responsible for at most two regions,
    so $p_{i_{j+2}}$ arrives in the same region as $q_2$, while it is unmatched.
    This is a contradiction, since \BIM would match them.
    Thus, other than $q_1$, the only never-matched point, $q_2$, in a region for which $e_k$ is responsible, arrives after $e_k$ and $w(q_2)<r w(p_{i_{k+1}})$.
   
    Then, for $k\geq 2$, the total weight of unmatched points in regions for which $e_k$ is responsible is at most 
    \[
        \left(r^{-(k-3)} + r^{-(k-1)} \frac{r(r-2)}{r-1} + \frac{r}{r-1} \right)\alpha + (1+r) \beta.
    \]
    If $r\leq\sqrt{2}$, then $r^{-(k-3)} + r^{-(k-1)} \frac{r(r-2)}{r-1}$ is at most zero and we can bound the total weight when $k\ge2$ by:
$(\frac{r}{r-1})\alpha + (1+r) \beta.$
    Thus, the ratio between the weight of matched points in the sequence $\langle p_{i_1}, p_{i_2}, \ldots, p_{i_{k+1}}\rangle$ and the total weight of all points associated with this sequence for $k \ge 2$ is at least $\frac{\alpha+\beta}{(\frac{r}{r-1})\alpha + (1+r) \beta}$. Since $\frac{r}{r-1}>1+r$, for $1<r\leq \sqrt{2}$, this ratio is minimized when $\beta$ is minimized, which happens at $\beta = r\alpha$. Thus, the competitive ratio for $k \ge 2$ is at least
    $(1+r)/(\frac{r}{r-1}+r(1+r))=\frac{r^2-1}{r^3}$.
    
    Now, consider the case of $k = 1$, and let $\alpha$ and $\beta$ have the same meaning as above. Then the sequence $\langle e_{i_1}, e_{i_2}, \ldots, e_{i_k}\rangle$ consists of a single edge. Thus, the weight of the matched points is $\alpha + \beta$, and there could be two unmatched points $q_1$ and $q_2$ at the end of the execution of the algorithm charged to this edge. We have $w(q_i) < r \beta$, so the ratio between the weight of matched points and the total weight of all points associated with the sequence for the case of $k = 1$  is at least $\frac{\alpha + \beta}{\alpha + (1+2r)\beta}$. Observe that this ratio is minimized when $\beta$ goes to infinity and becomes $1/(1+2r)$.

    Taking the worse ratio between the above two scenarios proves the statement of the theorem.
    \end{proof}

\begin{corollary}
  With the choice of parameter for \BIM, $r^*$, defined as the positive solution to the equation
  $\frac{1}{1+2r}=\frac{r^2-1}{r^3}$,
  approximately $1.2470$,
  we get a competitive ratio of
  $\frac{1}{1+2r^*}$, at least $0.2862$.
    \end{corollary}
\begin{proof}
  The value $r^*$ is obtained by setting the two terms in the minimum in Theorem~\ref{thm:revoke_ub} equal to each other and solving for $r$,
  giving the lower bound on the competitive ratio.
    
    To show that this result is tight for this algorithm, consider the following input: $p_1$ of weight $\alpha$ arrives at the north pole of the unit sphere, followed by $p_2$ of weight $\beta \ge \alpha$ at the south pole of the unit sphere, followed by $p_3$ of weight $r^* \beta-\varepsilon$ at the west pole of the unit sphere, and followed by $p_4$ of weight $r^* \beta - \varepsilon$ at the east pole of the unit sphere. The algorithm would end up matching $p_1$ with $p_2$, leaving $p_3$ and $p_4$ unmatched. Thus, in such an instance, the competitive ratio of the algorithm is $(\alpha + \beta)/(\alpha + \beta + 2r^*\beta - 2 \varepsilon)$. Taking $\beta$ to infinity and $\varepsilon$ to zero shows that the algorithm does not guarantee a competitive ratio better than $1/(1+2r^*)$ in the strict sense.
\end{proof}

\section{Collinear Points} \label{sec:collinear}

In this section, we investigate the non-crossing matching problem when all points lie on a line. First, we show that, contrary to the general case, randomization or revoking  alone does not yield any non-trivial competitive ratio, even when the points are unweighted. In other words, this shows that the results for the randomized algorithm and the revoking algorithm we presented above rely on points being in general position. 

\begin{theorem}	\label{thm:random_revoke_negative}
When points arrive in arbitrary positions, even if the points are unweighted and all lie on a line,

	(1) no randomized algorithm can achieve a competitive ratio better than $2/n$, and

	(2) no deterministic algorithm with revoking can achieve a competitive ratio better than $1/n$. 
\end{theorem}

\begin{proof}
   (1) We construct a random input, bound the competitive ratio of any deterministic algorithm on that input, and prove the theorem using Yao's minimax principle. To begin, place $p_1$ and $p_2$ in arbitrary positions on the line, then place $p_3$ between them. For each subsequent point $p_i$, randomly choose either the left or right segment relative to the last point $p_{i-1}$, and place $p_i$ in the middle of that segment. This process continues until a total of $2n$ points are positioned on the line.

  If a deterministic algorithm matches $p_1$ and $p_2$,
  the probability that it can match further points is zero.
  If it does \emph{not} match $p_1$ and $p_2$, we are in the following
  situation:
  Whenever the algorithm matches a new arriving point, there is a $\frac{1}{2}$ probability that all future points will arrive on the same side, preventing further matches.
  Consequently, the expected size of the matching created by the algorithm is bounded by $1+\frac12(1+\frac12(1+\frac12(1 + \cdots)))=2$.
However, a perfect matching is always possible, leading to a competitive ratio of $\frac{2}{n}$.

(2) The adversary’s strategy is designed to ensure that the algorithm can maintain at most one matched pair at any time. Initially, the adversary places points in arbitrary positions until the algorithm forms a match between two points. Once this match is created, the adversary starts placing new points between the matched pair, forcing the algorithm to revoke the match in order to form any new match.

Once the algorithm revokes the match, the adversary goes back to placing points in arbitrary positions. However, as soon as the algorithm creates another match, the adversary again places points between the matched pair, preventing any further progress. This cycle continues: the adversary alternates between placing points freely when no match exists and placing points between matched pairs once a match is formed. This ensures that the algorithm can only ever maintain one match, leading to a competitive ratio of $1/n$.
\end{proof}

Next, we present a randomized algorithm with revoking, called ``Random-Revoking-Matching'' (\RRM), which achieves a competitive ratio of at least $0.5$ for the unweighted case (see Algorithm \ref{alg:rrm}). \RRM maintains a partition of the line into intervals and matches points only within those intervals. The intervals may be open, half-open, or closed. Initially, the entire line is an open interval. Throughout the execution of the algorithm, the following conditions hold for partitioning the line:

\begin{itemize}
    \item Open and half-open intervals contain no matched pairs.
    \item Each half-open interval contains exactly one unmatched point, located on its closed boundary.
    \item Closed intervals contain exactly two points, both located on the boundaries and matched together.
\end{itemize}

These conditions hold at the beginning, with the line being one open interval, and they remain valid after each step of the algorithm. The decisions of the algorithm ensure that these conditions are preserved throughout its execution.

The first point that arrives in an open interval remains unmatched, and the partitioning does not change. When a second point arrives in an open interval, it is matched to the previously available point. This action splits the open interval into a closed interval containing the matched points and two new open intervals on either side of the matched segment.

When a point arrives in a closed interval, \RRM revokes the matched segment, randomly matches the new point to one of the now available points, and divides the interval into a closed interval containing the newly matched points and a half-open interval with the previously matched point on its closed boundary.

Finally, when a point arrives in a half-open interval, it is matched to the unmatched point on its closed boundary, splitting the interval into a closed interval between the matched points and an open empty interval (see Figure~\ref{fig:rrm}).

\begin{algorithm}[!h]
\caption{$RandomRevokingMatching$}\label{alg:rrm}
\begin{algorithmic}
\Procedure{$RandomRevokingMatching$}{}
\While{receive a new point $p_i$}
\State{Let $I$ be the interval that $p_i$ arrives in.}
\If{$I$ is an open interval $(x,y)$}
\If{$I$ contains another point $p_j$}
\State{Match $p_i$ with $p_j$.}
\If{$p_i$ is on the left side of $p_j$}
\State{Divide $I$ into $(x,p_i)$, $[\,p_i,p_j]\,$, and $(p_j,y)$.}
\Else
\State{Divide $I$ into $(x,p_j)$, $[\,p_j,p_i]\,$, and $(p_i,y)$.}
\EndIf
\EndIf
\ElsIf{$I$ is a closed interval $[\,p_j,p_k]\,$}
\State{Revoke $\overline{p_jp_k}$.}
\State{Set $r$ uniformly at random to $0$ or $1$.}
\If{$r = 0$}
\State{Match $p_i$ with $p_j$}
\State{Divide $I$ into $[\,p_j,p_i]\,$ and $(p_i,p_k]\,$}
\Else
\State{Match $p_i$ with $p_k$}
\State{Divide $I$ into $[\,p_j,p_i),$ and $[\,p_i,p_k]\,$}
\EndIf
\ElsIf{$I$ is a half-open interval $[\,p_j,x)$}
\State{Match $p_i$ with $p_j$.}
\State{Divide $I$ into $[\,p_j,p_i]\,$ and $(p_i,x)$}
\ElsIf{$I$ is a half-open interval $(x,p_j]\,$}
\State{Match $p_i$ with $p_j$.}
\State{Divide $I$ into $(x,p_i)$ and $[\,p_i,p_j]\,$}
\EndIf
\EndWhile
\EndProcedure
\end{algorithmic}
\end{algorithm}

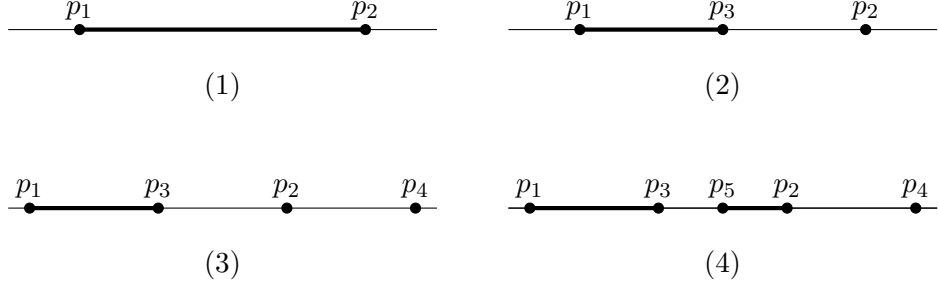
\begin{figure}[ht]
\centering
\vspace{0.5cm}
\begin{tikzpicture}[scale=0.95]

\draw[thin] (-3,0) -- (3,0);
\filldraw (-2,0) circle (2pt) node[above] {$p_1$};
\filldraw (2,0) circle (2pt) node[above] {$p_2$};
\draw[ultra thick] (-2,0) -- (2,0); 
\node at (0,-0.8) {$(1)$};

\begin{scope}[xshift=7cm]
\draw[thin] (-3,0) -- (3,0); 
\filldraw (-2,0) circle (2pt) node[above] {$p_1$};
\filldraw (0,0) circle (2pt) node[above] {$p_3$};
\filldraw (2,0) circle (2pt) node[above] {$p_2$};
\draw[ultra thick] (-2,0) -- (0,0); 
\node at (0,-0.8) {$(2)$};
\end{scope}

\begin{scope}[yshift=-2.5cm]
\draw[thin] (-3,0) -- (3,0); 
\filldraw (-2.7,0) circle (2pt) node[above] {$p_1$};
\filldraw (-0.9,0) circle (2pt) node[above] {$p_3$};
\filldraw (0.9,0) circle (2pt) node[above] {$p_2$};
\filldraw (2.7,0) circle (2pt) node[above] {$p_4$};
\draw[ultra thick] (-2.7,0) -- (-0.9,0); 
\node at (0,-0.8) {$(3)$};
\end{scope}

\begin{scope}[xshift=7cm, yshift=-2.5cm]
\draw[thin] (-3,0) -- (3,0); 
\draw[thin] (-3,0) -- (3,0); 
\filldraw (-2.7,0) circle (2pt) node[above] {$p_1$};
\filldraw (-0.9,0) circle (2pt) node[above] {$p_3$};
\filldraw (0,0) circle (2pt) node[above] {$p_5$};
\filldraw (0.9,0) circle (2pt) node[above] {$p_2$};
\filldraw (2.7,0) circle (2pt) node[above] {$p_4$};
\draw[ultra thick] (-2.7,0) -- (-0.9,0); 
\draw[ultra thick] (0,0) -- (0.9,0); 
\node at (0,-0.8) {$(4)$};
\end{scope}

\end{tikzpicture}
\caption{(1) The first two points arrive, partitioning the line into $(-\infty,p_1)$, $[\,p_1,p_2]\,$, and $(p_2, \infty)$. (2) Point $p_3$ arrives in $[\,p_1,p_2]\,$. \RRM revokes $\overline{p_1p_2}$, randomly matches $p_3$ to $p_1$, and partitions $[\,p_1,p_2]\,$ into $[\,p_1,p_3]\,$ and $(p_3,p_2]\,$. (3) Point $p_4$ arrives in $(p_2,\infty)$ and remains unmatched as the interval it arrived in is empty. (4) Point $p_5$ arrives in $(p_3,p_2]\,$ and is matched to $p_2$, splitting the interval into $(p_3,p_5)$ and $[\,p_5,p_2]\,$.}
\label{fig:rrm}
\end{figure}

\begin{theorem}
    When points are unweighted and arrive on a line,
    the competitive ratio of the randomized algorithm \RRM, with access to revoking, is at least $1/2$.
\end{theorem}

\begin{proof}
    Let $f(n)$, $g(n)$, and $h(n)$ represent the expected number of matched pairs when $n$ points arrive in a closed, half-open, and open interval, respectively. For $n=0$, there are no matched pairs. By the definition of \RRM, when a point arrives in an open interval, it remains unmatched, and when a point arrives in a closed interval, a matched pair is revoked and a new one is created, so the number of matched pairs does not change. However, when a point arrives in a half-open interval, \RRM creates a new matching. Therefore, we have the initial conditions: $f(0) = f(1) = g(0) = h(0) = h(1) = 0$ and $g(1) = 1$.

    If two or more points arrive in an open interval, the first two points will be matched, dividing the interval into two open intervals and one closed interval. When $n$ points are expected to arrive in the original open interval, the new intervals will receive $k_1$, $k_2$, and $n - k_1 - k_2 - 2$ points, where $k_1, k_2 \geq 0$ and $k_1 + k_2 \leq n - 2$. Therefore, for $n \geq 2$, we have:
\[
h(n) \geq 1 + \min_{\substack{k_1, k_2 \\ k_1, k_2 \geq 0 \\ k_1 + k_2 \leq n - 2}} \{h(k_1) + f(k_2) + h(n - k_1 - k_2 - 2)\}.
\]

    The first point that arrives in a closed interval revokes the existing matched segment and is randomly matched to one of the points, dividing the interval into one half-open and one closed interval. When $n \geq 2$ points are expected to arrive, the new intervals will receive $k$ and $n - k - 1$ points, where $0 \leq k \leq n - 1$. Therefore, for $n \geq 2$, we have:
\[
f(n) \geq \min_{0 \leq k \leq n - 1} \left\{ \frac{1}{2}\left[ f(k) + g(n - k - 1) \right] + \frac{1}{2}\left[ f(n - k - 1) + g(k) \right] \right\}.
\]

Finally, when the first point arrives in a half-open interval, it gets matched to the only available point, dividing the interval into a closed and an open interval. When $n \geq 2$ points are expected to arrive, the new intervals will receive $k$ and $n - k - 1$ points, where $0 \leq k \leq n - 1$. Therefore, for $n \geq 2$, we have:
\[
g(n) \geq 1 + \min_{0 \leq k \leq n - 1} \left\{ h(k) + f(n - k - 1) \right\}.
\]

   Next, we show the following inequalities by induction for $n \geq 0$:
\[
f(n) \geq \frac{n-1}{4}, \qquad
g(n) \geq \frac{n+1}{4}, \qquad
h(n) \geq \frac{n-1}{4}.
\]

It is straightforward to verify these inequalities for $n = 0$ and $n = 1$. For the induction step, for any $k, k_1, k_2$ such that $0 \leq k \leq n-1$, $k_1, k_2 \geq 0$, and $k_1 + k_2 \leq n - 2$, we have:
\begin{align*}
	h(n) &\geq 1 + h(k_1) + f(k_2) + h(n - k_1 - k_2 - 2)	\\
		&\geq 1 + \frac{k_1 - 1}{4} + \frac{k_2 - 1}{4} + \frac{n - k_1 - k_2 - 3}{4} = \frac{n - 1}{4}.
\end{align*}

For $f(n)$, we get:
\begin{align*}
	f(n) &\geq \frac{1}{2} \left( f(k) + g(n - k - 1) \right) + \frac{1}{2} \left( f(n - k - 1) + g(k) \right) \\
		&\geq \frac{1}{2} \left( \frac{k - 1}{4} + \frac{n - k}{4} + \frac{n - k - 2}{4} + \frac{k + 1}{4} \right) = \frac{n - 1}{4}.
\end{align*}

For $g(n)$, we have:
\[
g(n) \geq 1 + h(k) + f(n - k - 1) \geq 1 + \frac{k - 1}{4} + \frac{n - k - 2}{4} = \frac{n + 1}{4}.
\]

Thus, the expected number of matched pairs by \RRM when $2n$ points arrive on a line is at least
$h(2n) \ge \frac{2n - 1}{4}$,
where the optimum is always $n$, which implies that \RRM achieves a competitive ratio of $1/2$.
\end{proof}

\section{Algorithms with Advice}
\label{sec:advice}

In this section, we consider the advice complexity of the \WNM problem, i.e., how many perfect bits of information is sufficient to be optimal. We use the tape advice model, which is the standard advice model. In this setting, a trustworthy oracle cooperates with an online algorithm according to a pre-agreed protocol. The oracle has access to the entire input sequence in advance. The oracle communicates with an online algorithm by writing bits on the advice tape. When an input item arrives, an algorithm reads some number of advice bits from the tape and makes a decision for the new item based on the advice it read from the tape so far, and the items that have arrived so far. The question of interest is to bound the number of bits that need to be communicated between the oracle and the algorithm on the worst-case input to allow an online algorithm to solve the problem optimally. For an introduction to online algorithms with advice, an interested reader is referred to the survey~\cite{boyar2017online} and references therein.

We propose an online algorithm with advice, which we call ``Split-And-Match'' (\SAM), and show that it achieves optimality. For input sequences of size $2n$, \SAM uses a family of $C_n$ advice strings, where $C_n$ is the $n^\text{th}$ Catalan number.
There are many alternative definitions of the Catalan numbers. Perhaps the most relevant here is $C_0=1, C_n=\frac{2(2n-1)}{n+1}C_{n-1}$, since it most clearly indicates how $C_n$ grows asymptotically.
Note that $O(\log C_n) = O(n)$, and, thus, $O(\log\log C_n) = O(\log n)$.
The oracle encodes each advice string using the Elias delta coding scheme~\cite{Elias75}, which requires
$\log C_n + O(\log\log C_n)$ bits.

The \SAM oracle and algorithm jointly maintain a partitioning of the plane into convex regions and a responsibility relation, where a point can be assigned to be responsible for at most one region, and each region can have at most one point responsible for it. Each region defined will eventually receive an even number of points in total. No region which has not yet been divided into sub-regions contains more than one unmatched point. When a new point, $p$, arrives in a region, $R$, if $R$ does not have a responsible point, then $p$ is assigned to be the responsible point for $R$, and $p$ is left unmatched at this time. Otherwise, suppose that $q$ is the responsible point for $R$ at the time $p$ arrived. In this case, the responsibility of $q$ is removed, and the plane partition is refined by subdividing $R$ into $R_1$ and $R_2$, the sub-regions of $R$ formed by $\overline{pq}$. If the total number of points (including future points, but excluding $p$ and $q$) in $R_i$ is even for each $i \in \{1,2\}$, then $p$ and $q$ are matched (we refer to this event as a ``safe match''), and no points are responsible for either $R_1$ or $R_2$. Otherwise, $p$ and $q$ are not matched, and $q$ is made responsible for $R_1$, and $p$ is made responsible for $R_2$. Note that when a region has a responsible point, that point is assumed to lie in the region by convention, though it can lie on the boundary.

To implement the above procedure in the advice model, the \SAM oracle 
(see Algorithm~\ref{alg:split_oracle}) 
creates a binary string $D$ of length $2n$, where the $i^\text{th}$ bit indicates whether $p_i$ arrives in a region which has some responsible point $p_j$ assigned to it, and $p_i$ and $p_j$ form a safe match. The string $D$ is encoded  on the tape and is passed to \SAM. The \SAM algorithm (see Algorithm~\ref{alg:asap}) 
reads the encoding of $D$ from the tape (prior to the arrival of online points), recovers $D$ from the encoding, and then uses the information in $D$ to run the above procedure creating safe matches.

\begin{algorithm}[h]
\caption{$Split$-$And$-$Match$ Oracle.}\label{alg:split_oracle}
\begin{algorithmic}
\Procedure{$Split$-$And$-$Match$-$Oracle$}{}
\State{$D \gets [0]$}
\State{make $p_1$ responsible for the plane}
\For{$i=2$ to $2n$}
    \State{let $R$ be the region that $p_i$ arrives in}
    \If{$R$ has a responsible point $p_j$}
        \State{revoke the responsibility of $p_j$}
        \State{divide $R$ into $R_1$ and $R_2$ by $\overline{p_ip_j}$}
        \If{$R_L$ (and $R_R$) is going to contain an even number of points in total}
            \State{$D.append(1)$}
        \Else
            \State{make $p_j$ and $p_i$ responsible for $R_1$ and $R_2$ respectively}
            \State{$D.append(0)$}
        \EndIf
    \Else
        \State{make $p_i$ responsible for $R$}
        \State{$D.append(0)$}
    \EndIf
\EndFor
\State{pass $D$ to the algorithm}
\EndProcedure\\
\end{algorithmic}
\end{algorithm}

\begin{algorithm}[h]
\caption{$Split$-$And$-$Match$ Algorithm.}\label{alg:asap}
\begin{algorithmic}

\Procedure{$Split$-$And$-$Match$}{$D$}
\While{receive a new point $p_i$}
    \State{let $R$ be the region that $p_i$ arrives in}
    \If{$R$ has a responsible point $p_j$}
        \State{revoke the responsibility of $p_j$}
        \State{divide $R$ into $R_1$ and $R_2$ by $\overline{p_ip_j}$}
        \If{$D[i] == 1$}
            \State{match $p_i$ with $p_j$}
        \Else
            \State{make $p_j$ and $p_i$ responsible for $R_1$ and $R_2$ respectively}
            \State{leave $p_i$ unmatched}
        \EndIf
    \Else
        \State{make $p_i$ responsible for $R$}
        \State{leave $p_i$ unmatched}
    \EndIf
\EndWhile
\EndProcedure
\end{algorithmic}
\end{algorithm}

Not all binary strings of length $2n$ can be generated as a valid $D$. Thus, the oracle and the algorithm can agree beforehand on the ordering of the universe of possible strings $D$, which we call the advice family. Then the oracle writes on the tape the index of a string in this ordering that corresponds to $D$ for the given input.
The following theorem establishes the correctness of this algorithm, as well as the claimed bound on the advice complexity.

\begin{theorem}
\label{thm:sam}
    \SAM achieves a perfect matching with the advice family of size $C_n$.
\end{theorem}
\begin{proof}
    Since the request sequence contains $2n$ points, an even number of points eventually arrive. Inductively, a region that is divided always has an even number of points in both sub-regions, and each of these sub-regions is convex. Thus, any point arriving in a region can be matched to the point that is already there and responsible for the region.
    There are only two kinds of regions that occur during the execution of \SAM, as a new point arrives:
    \begin{itemize}
        \item type I: this region does not have a responsible point, it is empty at the time of creation, and there are an even number of points arriving in this region in the future, and
        \item type II: this region has a responsible point, which is the only point in the region at the time of its creation, and there are an odd number of points arriving in this region in the future.
    \end{itemize}
    We argue inductively (on the number of future points arriving in a region) that the algorithm ends up matching all points inside a region, regardless of their type.
    The base case for a type I region is trivial: the number of future points is $0$, and there is nothing to prove. The base case for type II region is easy: one point arrives in the region, then according to the algorithm it will be matched to the responsible point (since $R_1$ and $R_2$ are empty).

    For the inductive step, consider a type I region $R$, and suppose that $2k$ points will arrive inside the region. The first point that arrives in the region becomes responsible for this region, changing its type to II. There are $2k-1$ future points arriving in this region, and the claim follows by the inductive assumption applied to the type II region.
    Now, consider a type II region $R$, and suppose that $2k-1$ points arrive inside the region. Let $q$ be the responsible point for $R$, and let $p$ be the first point arriving inside $R$. Note that $\overline{pq}$ partitions $R$ into $R_1$ and $R_2$ and an even number of points will arrive in $R_1$ and $R_2$ in total. There are two possible cases. Case 1: If $R_1$ and $R_2$ are both of type I, then $p$ is matched with $q$. Case 2: If $R_1$ and $R_2$ are both of type II, then the match is not formed. In both
    cases, the inductive step is established for $R$ by invoking induction on $R_1$ and $R_2$.

    Observe that the entire plane is a region of type I at the beginning of the execution of the algorithm (prior to arrival of any points). Thus, the correctness of the algorithm follows by applying the above claim to this region.

    To establish the bound on advice complexity, observe that by the definition of the algorithm, \SAM matches the most recent point whenever $D[i]$ is $1$ and does not match otherwise. Thus, $D$ has an equal number of zeros and ones and no prefix of $D$ has more ones than zeros. This makes $D$ a Dyck word and it is known that there are $C_n$ Dyck words of size $2n$~\cite{Roman15}.
\end{proof}

\section{Conclusion}
We introduced the weighted version of the Online Weighted Non-Crossing Matching problem. We established that no deterministic algorithm can guarantee a constant competitive ratio for this problem. Then, we explored several ways of overcoming this limitation and presented new algorithms and bounds for each of the considered regimes. In particular, we presented the results for deterministic algorithms when weights of the points are restricted to lie in the range $[1, U]$, randomized algorithms without restrictions on weights, deterministic algorithms with revoking, and deterministic algorithms with advice. 

Many open problems remain. In particular, our bounds are not tight, and closing the gap in any of the settings would be of interest. It is also interesting to study the online setting of other versions of the problem that were considered in the offline literature. For example, one could allow an online algorithm to create some number of crossings up to a given budget, or one could consider the $k$-non-crossing constraint as inspired from understanding RNA structures. In this paper, we considered the vertex-weighted version, but one could also consider an edge-weighted version of the problem, where edge weights could be either abstract or related to geometry.

\bibliography{refs.bib}
\bibliographystyle{plain}

\end{document}